\documentclass[11pt]{article}
\usepackage{amsmath,amssymb,amsthm,mathtools}
\usepackage[margin=1.15in]{geometry}
\newtheorem{theorem}{Theorem}[section]
\newtheorem{lemma}[theorem]{Lemma}
\newtheorem{corollary}[theorem]{Corollary}
\theoremstyle{definition}
\newtheorem{example}[theorem]{Example}
\theoremstyle{remark}\newtheorem{remark}[theorem]{Remark}
\DeclareMathOperator{\Tr}{Tr}
\newcommand{\M}{\mathcal M}
\title{Extremal Local Determinant Defects in Semifinite von Neumann Algebras}
\author{S. Mahmoud Manjegani\thanks{Department of Mathematics and Statistics, University of Regina, Regina, SK S4S 0A2, Canada; Department of Mathematical Sciences, Isfahan University of Technology, Isfahan 84156-83111, Iran. Email: \texttt{Seyed.Manjegani@uregina.ca}.}}
\date{}
\begin{document}
\maketitle
\begin{abstract}
Let $(\M,\tau)$ be a semifinite von Neumann algebra with a faithful normal semifinite trace. We study the local determinant defect of products of positive measurable operators with prescribed upper spectral projections and spectral gaps. We obtain an exact formula for the defect of two-level operators in terms of the spectral distribution of the compression of one projection to the range of the other. We show that these operators minimize the defect in the corresponding class of measurable operators. For bounded positive invertible operators, under a spectral gap assumption, we characterize all minimizers. Under an additional geometric assumption and a uniform norm bound, we obtain a quadratic estimate for the distance to the set of minimizers and show that the exponent two is optimal. As an application, we obtain a bound for the distance between low-energy spectral projections of Gibbs states.
\end{abstract}



\section{Introduction}
Generalized singular values are useful in the study of inequalities in semifinite von Neumann algebras; see \cite{FackKosaki}. The submultiplicative property of the determinant function for measurable operators was established in \cite[Theorem 4.2(ii)]{DoddsEtAl}.

Let $(\M,\tau)$ be a semifinite von Neumann algebra with a faithful normal semifinite trace. For a measurable operator $X$, whenever the integral is well defined, put
\[
\ell_t(X)=\int_0^t\log\mu_s(X)\,ds.
\]
For operators in the logarithmic class defined in Section 2, the determinant inequality gives
\[
\ell_t(AB)\le\ell_t(A)+\ell_t(B).
\]
When the two terms on the right are finite, we define the local determinant defect by
\[
\delta_t(A,B)=\ell_t(A)+\ell_t(B)-\ell_t(AB).
\]

The determinant inequality gives the non-negativity of the defect. Our aim is to determine its minimum under prescribed spectral constraints and to study the structure and stability of the minimizers. We consider positive measurable operators with prescribed upper spectral projections and spectral gaps. Let $P,Q$ be projections with $0<\tau(P)=\tau(Q)=t<\infty$ and let $0<p,q<1$. We consider positive operators $A,B$ in the logarithmic class such that $P$ reduces $A$, $Q$ reduces $B$, and
\[
AP\ge P,\quad AP^\perp\le pP^\perp,\qquad
BQ\ge Q,\quad BQ^\perp\le qQ^\perp.
\]
We show that the minimum is attained at the two-level operators
\[
A_0=P+pP^\perp,\qquad B_0=Q+qQ^\perp.
\]
Using the classical decomposition of two projections \cite{HalmosTwoSubspaces}, we compute the minimum in terms of the spectral distribution of $PQP$ in $P\M P$. The formula also gives a lower bound for the defect in terms of the distance between $P$ and $Q$. The minimum theorem follows from the determinant inequality in \cite{DoddsEtAl}.

We then study the equality case in the minimum problem. For bounded positive invertible operators, under a spectral gap assumption on $|A_0B_0|$, we characterize all minimizers. Under a further assumption on the relative position of the relevant projections and a uniform norm bound, we obtain an explicit quadratic estimate for the distance to the set of minimizers. A three-dimensional example shows that the exponent two cannot in general be replaced by a smaller positive exponent.

Section 2 recalls some facts about measurable operators and local determinants. In Section 3, we compute the exact defect for two-level operators. Section 4 contains the minimum theorem for measurable operators. In Section 5, we characterize the minimizers and establish a quantitative stability estimate. Section 6 gives an application to low-energy spectral projections of Gibbs states.

\section{Preliminaries on measurable operators and local determinants}
Let $(\M,\tau)$ be a semifinite von Neumann algebra with a faithful normal semifinite trace. We denote by $S(\M,\tau)$ the algebra of $\tau$-measurable operators affiliated with $\M$, and by $\mu_s(X)$ the generalized singular values of $X$; see \cite{FackKosaki}. Products of measurable operators are understood as products in this algebra.

Following \cite{DoddsEtAl}, let
$L_{\log_+}(\M,\tau)=\{X\in S(\M,\tau):\log_+|X|\in L^1(\M,\tau)+\M\}$.
For $X$ in this class and $0<t<\infty$, set
\[
\ell_t(X)=\int_0^t\log\mu_s(X)\,ds,
\]
where $\log 0=-\infty$. The integral takes values in $[-\infty,\infty)$. We use the local determinant inequality \cite[Theorem 4.2(ii)]{DoddsEtAl}
\begin{equation}\label{eq:submult}
\ell_t(XY)\leq\ell_t(X)+\ell_t(Y).
\end{equation}
The class $L_{\log_+}(\M,\tau)$ is closed under multiplication. If $\ell_t(X)$ and $\ell_t(Y)$ are finite, define
\[
\delta_t(X,Y)=\ell_t(X)+\ell_t(Y)-\ell_t(XY).
\]
If $\ell_t(XY)=-\infty$, we set $\delta_t(X,Y)=+\infty$. By \eqref{eq:submult}, the defect is nonnegative.

For projections $P,Q$ with
$0<\tau(P)=\tau(Q)=t<\infty$,
let $\nu_{P,Q}$ be the spectral distribution
of $PQP$ in $P\M P$ with respect to the
normalized trace $\tau/t$. That is,
\[
\nu_{P,Q}(B)=\frac{1}{t}\tau(E_{PQP}(B))
\]
for every Borel set $B\subseteq[0,1]$,
where $E_{PQP}$ is the spectral measure
of $PQP$ in $P\M P$.
Thus, for every bounded Borel function
$u$ on $[0,1]$,
\[
t\int_{[0,1]}u(x)\,d\nu_{P,Q}(x)
=\tau(u(PQP)),
\]
where the functional calculus on the right
is taken in $P\M P$.
We write
$\|X\|_{2,\tau}=\tau(X^*X)^{1/2}$,
allowing the value $+\infty$.

\section{Exact local determinant defects for two projections}
\begin{lemma}[Tracial form of the two-projection decomposition]
\label{lem:tracial-two-projection-decomposition}
Let $e,f\in\M$ be projections satisfying
$0<\tau(e)=\tau(f)=t<\infty$,
and set
\[
e_{11}:=e\wedge f,
\qquad
e_{10}:=e\wedge f^\perp,
\qquad
e_{01}:=e^\perp\wedge f,
\qquad
e_{00}:=e^\perp\wedge f^\perp.
\]
After removing these four reducing subspaces, the generic part of
the von Neumann algebra generated by $e$ and $f$ is
trace-preservingly isomorphic to $M_2(\mathcal A)$ for an
abelian von Neumann algebra $\mathcal A$. Under this identification,
\begin{equation}\label{eq:generic-two-projection-model}
e=
\begin{pmatrix}
1&0\\
0&0
\end{pmatrix},
\qquad
f=
\begin{pmatrix}
c&\sqrt{c(1-c)}\\
\sqrt{c(1-c)}&1-c
\end{pmatrix},
\end{equation}
where $c\in\mathcal A$ is a positive contraction satisfying
$\chi_{\{0\}}(c)=\chi_{\{1\}}(c)=0$.
The restriction of $\tau$ to the generic part has the form
$\varphi\otimes\operatorname{Tr}_2$,
where $\varphi$ is a finite normal trace on $\mathcal A$. Moreover,
\begin{equation*}
\tau(e_{10})=\tau(e_{01}).
\end{equation*}
If $\nu_{e,f}$ denotes the normalized spectral distribution of
$efe$ in the corner $e\M e$, then, for every bounded Borel function
$u$ on $[0,1]$,
\begin{equation}\label{eq:angle-measure-decomposition}
t\int_{[0,1]}u(x)\,d\nu_{e,f}(x)
=
\tau(e_{11})u(1)
+\tau(e_{10})u(0)
+\varphi(u(c)).
\end{equation}
\end{lemma}

\begin{proof}
Let $g:=e\vee f$. Then
$\tau(g)\leq\tau(e)+\tau(f)=2t$,
so $g\M g$ is a finite corner. Apply the standard two-projection
decomposition \cite{HalmosTwoSubspaces} in this corner. On the
generic part, the polar component of $(1-e)fe$ identifies the two
diagonal corners. With respect to this identification, $efe$ is
represented by $c$, and the projection relations give
\eqref{eq:generic-two-projection-model}.

The two diagonal corners carry the same finite normal trace
$\varphi$. Hence the generic contributions to $\tau(e)$ and
$\tau(f)$ are both equal to $\varphi(1)$. On the remaining reducing
summands,
\[
\tau(e)=\tau(e_{11})+\tau(e_{10})+\varphi(1)
\quad \text{and} \quad
\tau(f)=\tau(e_{11})+\tau(e_{01})+\varphi(1).
\]
Since $\tau(e)=\tau(f)$, it follows that
$\tau(e_{10})=\tau(e_{01})$.

Finally, on the $e$-corner, the operator $efe$ is equal to $1$ on
$e_{11}$, to $0$ on $e_{10}$, and to $c$ on the generic part.
Therefore,
$\tau\bigl(u(efe)\bigr)=
\tau(e_{11})u(1)
+\tau(e_{10})u(0)
+\varphi(u(c))$.
By the definition of $\nu_{e,f}$,
\[
\tau\bigl(u(efe)\bigr)
=
t\int_{[0,1]}u(x)\,d\nu_{e,f}(x),
\]
which proves \eqref{eq:angle-measure-decomposition}.
\end{proof}

For $0<p,q<1$, set $A_0=P+pP^\perp$ and  $B_0=Q+qQ^\perp$.
For $x\in[0,1]$, define
\begin{equation}\label{eq:eta}
\eta_{p,q}(x)=\frac{F(x)+\sqrt{F(x)^2-4p^2q^2}}{2},
\end{equation}
where $F(x)=p^2+q^2+(1-p^2)(1-q^2)x$. Set
\[
C(p,q)=\frac{(1-p^2)(1-q^2)}{1-p^2q^2}.
\]
\begin{theorem}\label{thm:exact}
Let $P,Q\in\M$ be projections with $0<\tau(P)=\tau(Q)=t<\infty$. Then
\begin{equation}\label{eq:exact}
\delta_t(A_0,B_0)=-\frac{t}{2}\int_{[0,1]}\log\eta_{p,q}(x)\,d\nu_{P,Q}(x).
\end{equation}
Moreover,
\begin{equation}\label{eq:distance}
\delta_t(A_0,B_0)\geq\frac{C(p,q)}4\|P-Q\|_{2,\tau}^2.
\end{equation}
In particular, $\delta_t(A_0,B_0)=0$ if and only if $P=Q$.
\end{theorem}
\begin{proof}
Since $A_0$ and $B_0$ equal $1$ on projections of trace $t$, we have $\ell_t(A_0)=\ell_t(B_0)=0$. Apply Lemma~\ref{lem:tracial-two-projection-decomposition} to $P,Q$. On a generic two-dimensional fiber, the squared singular values of $A_0B_0$ are the roots of
$\lambda^2-F(x)\lambda+p^2q^2=0$.
They are $\eta=\eta_{p,q}(x)$ and $\eta_-=p^2q^2/\eta$. Since $F$ is increasing,
\[
\sqrt{\eta_-}\leq\min\{p,q\}\leq\max\{p,q\}\leq\sqrt\eta.
\]
On the common intersection the singular value is $1$. On the two mismatch summands the values are $q$ and $p$, and these summands have equal trace. On $(P\vee Q)^\perp$ the value is $pq$. Hence the upper $t$ units of the singular-value distribution consist of the common intersection, the larger mismatch value and the upper generic branch. This remains true when $p=q$, since tied values give the same integral. Formula \eqref{eq:angle-measure-decomposition} gives \eqref{eq:exact}.

Since $\eta\leq1$ and $\eta_-\geq p^2q^2$, the identity
$(1-\eta)(1-\eta_-)=(1-p^2)(1-q^2)(1-x)$
implies $-\log\eta\geq1-\eta\geq C(p,q)(1-x)$. Finally,
\[
t\int_{[0,1]}(1-x)\,d\nu_{P,Q}(x)=\tfrac12\|P-Q\|_{2,\tau}^2.
\]
This proves \eqref{eq:distance}. The equality assertion follows from faithfulness of the trace and $\eta_{p,q}(1)=1$.
\end{proof}

\section{Spectral flattening for measurable operators}
Fix $P,Q,p,q,t$ as in Theorem~\ref{thm:exact}. Let $\mathcal C$ be the class of pairs of positive operators $A,B\in L_{\log_+}(\M,\tau)$ such that $P$ reduces $A$, $Q$ reduces $B$, and
\begin{equation}\label{eq:admissible}
AP\geq P,\quad AP^\perp\leq pP^\perp,\qquad BQ\geq Q,\quad BQ^\perp\leq qQ^\perp.
\end{equation}
The inequalities are understood in the respective corners. In particular, the upper local logarithmic integrals of both factors are finite. No injectivity on the complementary corners is required.

\begin{theorem}\label{thm:flatten}
For every $(A,B)\in\mathcal C$,
\begin{equation}\label{eq:flatten}
\delta_t(A,B)\geq\delta_t(A_0,B_0).
\end{equation}
Consequently,
\[
\min_{(A,B)\in\mathcal C}\delta_t(A,B)
=-\frac{t}{2}\int_{[0,1]}\log\eta_{p,q}(x)\,d\nu_{P,Q}(x).
\]
In particular, $\delta_t(A,B)\geq C(p,q)\|P-Q\|_{2,\tau}^2/4$, and $\delta_t(A,B)=0$ if and only if $P=Q$.
\end{theorem}
\begin{proof}
Put $D=AA_0^{-1}$ and $E=B_0^{-1}B$. Both operators are positive and belong to $L_{\log_+}(\M,\tau)$. On $P$, $D$ agrees with $A$, and on $P^\perp$ it is bounded above by $1$. Similarly, $E$ agrees with $B$ on $Q$ and is bounded above by $1$ on $Q^\perp$. Hence
$\ell_t(D)=\ell_t(A)$ and $ \ell_t(E)=\ell_t(B)$.
Since $AB=D(A_0B_0)E$, two applications of \eqref{eq:submult} give
\[
\ell_t(AB)\leq\ell_t(A)+\ell_t(A_0B_0)+\ell_t(B).
\]
If the left side is $-\infty$, \eqref{eq:flatten} is immediate. Otherwise subtraction gives \eqref{eq:flatten}, since $\ell_t(A_0)=\ell_t(B_0)=0$. The pair $(A_0,B_0)$ belongs to $\mathcal C$, so the minimum is attained. The lower bound follows from Theorem~\ref{thm:exact}, and it implies that zero defect requires $P=Q$.

Conversely, suppose $P=Q$. In the finite corner $P\M P$, write $a=AP$ and $b=BP$. Both are bounded below by the identity of the corner, and their inverses are contractions. The product $ab$ has inverse $b^{-1}a^{-1}$, also a contraction. Thus $|ab|\ge P$. The complementary product has norm at most $pq$. It follows that the upper local logarithmic integrals of $A,B,AB$ are the full logarithmic integrals of $a,b,ab$ in the corner.

Write $L(x)=\int_0^t\log\mu_s(x)\,ds$ in this corner. The determinant inequality gives $L(ab)\le L(a)+L(b)$, so $L(ab)$ is finite since $|ab|\ge P$. By spectral calculus, $L(b^{-1})=-L(b)$. Applying the determinant inequality in the corner to $a=(ab)b^{-1}$ gives
\[
L(a)\le L(ab)+L(b^{-1})=L(ab)-L(b).
\]
Hence $L(ab)=L(a)+L(b)$, and $\delta_t(A,B)=0$.
\end{proof}

\begin{remark}
The normalization in \eqref{eq:admissible} entails no loss of generality. If the upper and lower boundary bounds of $A$ are $a_+>a_->0$, and those of $B$ are $b_+>b_->0$, replace $A,B$ by $A/a_+,B/b_+$ and take $p=a_-/a_+$ and $q=b_-/b_+$. The defect is unchanged under multiplication of either factor by a positive scalar.
\end{remark}

\section{Equality and quantitative stability}
In this section we consider the subclass $\mathcal C_b$ of $\mathcal C$ consisting of bounded positive invertible operators. Put $T=A_0B_0$ and write
\[
D=AA_0^{-1}=I+U-V,\qquad E=B_0^{-1}B=I+W-Z,
\]
where $U=PUP$, $V=P^\perp VP^\perp$, $W=QWQ$, and $Z=Q^\perp ZQ^\perp$ are positive. In particular, $\|V\|,\|Z\|<1$. The excess defect is
\begin{equation}\label{eq:excess}
\varepsilon=\delta_t(A,B)-\delta_t(A_0,B_0)
=\ell_t(D)+\ell_t(T)+\ell_t(E)-\ell_t(DTE).
\end{equation}
Assume that there is a spectral projection $R$ of $|T|$ with $\tau(R)=t$ and numbers $a>b>0$ such that
\begin{equation}\label{eq:gap}
|T|R\ge aR,\qquad |T|R^\perp\le bR^\perp.
\end{equation}
Let $L$ be the range projection of $TR$. By the polar decomposition, $\tau(L)=t$. The two-projection decomposition in Section 3 shows that \eqref{eq:gap} holds whenever $p\ne q$, with $a=\max(p,q)$ and $b=\min(p,q)$. If $p=q$, the condition $PQP\ge \eta P$ for some $\eta>0$ is also sufficient, by the separation of the two branches in \eqref{eq:eta}. We take $a\le1$ and put $g=a^2-b^2$.

\begin{lemma}\label{lem:volume}
Let $X$ be bounded and invertible and let $S$ be a projection of trace $t$. Define
\[
v_X(S)=\tfrac12\tau(\log_S(SX^*XS)),
\]
where the logarithm is taken in $S\M S$. Then $v_X(S)\le\ell_t(X)$. If $S_X$ denotes the range projection of $XS$, then for bounded invertible $Y$,
\begin{equation}\label{eq:volumeadd}
v_{YX}(S)=v_X(S)+v_Y(S_X).
\end{equation}
If $F$ is an upper spectral projection of $|X|$ of trace $t$ (allowing a subprojection of a boundary eigenspace), then $v_X(F)=\ell_t(X)$ and
\begin{equation}\label{eq:volumeloss}
\ell_t(X)-v_X(S)\ge
\frac{\tau(|X|^2F)-\tau(|X|^2S)}{2\|X\|^2}.
\end{equation}
\end{lemma}
\begin{proof}
For $0<s<t$, the singular values of $SX^*XS$ in the corner are bounded above by those of $X^*X$, by the compression inequality for generalized singular values \cite{FackKosaki}. Integration of their logarithms proves the first assertion. For an upper spectral projection $F$, the distribution in the corner is precisely the upper $t$ part of the distribution of $|X|^2$. Applying $\log u-\log v\ge (u-v)/\|X\|^2$ for $0<v\le u\le\|X\|^2$ to the same singular values gives \eqref{eq:volumeloss}.

Write $XS=K H$, where $H=(SX^*XS)^{1/2}$ is invertible in the corner and $K^*K=S$, $KK^*=S_X$. Then
$SX^*Y^*YXS=H K^*Y^*YK H$.
The multiplicativity of the Fuglede--Kadison determinant in the finite corner \cite{FugledeKadison}, and trace preservation under the corner isomorphism induced by $K$, give \eqref{eq:volumeadd}.
\end{proof}

An upper spectral projection of trace $t$ need not exist in general, even in a diffuse algebra. For example, a bounded operator in an infinite trace algebra may have a spectral supremum that is not an eigenvalue. We use approximate maximizers instead. Pass, if necessary, to $\M\overline\otimes L^\infty(0,1)$ with the product trace, replacing each operator by its tensor product with $1$. Singular values are unchanged. For the spectral projections of $H\otimes1$ used below, a subprojection of any required finite trace can be obtained by tensoring a finite projection with a measurable subset of $(0,1)$; semifiniteness supplies a sufficiently large finite subprojection if needed.

If $t=\tau(I)<\infty$, take $S=I$, which gives $v_X(S)=\ell_t(X)$. Otherwise, for a bounded invertible $X$, put $H=|X|$ and $d=\mu_t(X)>0$. The projection $1_{(d,\infty)}(H)$ has trace at most $t$, while $1_{(d-\xi,\infty)}(H)$ has trace at least $t$ for $0<\xi<d$. Extend the first projection to a projection $S$ of trace $t$ inside the second one. On the added part the compression of $H^2$ is at least $(d-\xi)^2$; the first part reduces $H$. Consequently
\[
0\le\ell_t(X)-v_X(S)\le t\log\frac{d}{d-\xi}.
\]
Thus for every $r>0$ we can choose $S$ with $v_X(S)\ge\ell_t(X)-r$. All estimates involving only the original operators descend to $\M$. Throughout, traces of products of positive operators are understood as traces of the corresponding positive sandwiches.

\begin{lemma}\label{lem:leak}
Suppose $\|D\|,\|E\|\le M$, where $M\ge1$. Under \eqref{eq:gap},
\begin{equation}\label{eq:leak}
\tau(U(I-L))+\tau(VL)+\tau(W(I-R))+\tau(ZR)
\le \left(8M^2+\frac{32M}{p^2q^2g}\right)\varepsilon.
\end{equation}
\end{lemma}
\begin{proof}
For $r>0$, choose a projection $S$ of trace $t$ with $v_{DTE}(S)\ge\ell_t(DTE)-r$, using the enlargement if necessary. Put $e_r=\ell_t(D)+\ell_t(T)+\ell_t(E)-v_{DTE}(S)$, so $\varepsilon\le e_r\le\varepsilon+r$. Let $S_1$ and $S_2$ be the range projections of $ES$ and $TES$. Both have trace $t$. By \eqref{eq:volumeadd} and \eqref{eq:excess},
\begin{equation*}
e_r=(\ell_t(E)-v_E(S))+(\ell_t(T)-v_T(S_1))+(\ell_t(D)-v_D(S_2)).
\end{equation*}
Each term is nonnegative. Applying \eqref{eq:volumeloss} to $T$ and using $\|T\|\le1$ and \eqref{eq:gap}, we obtain
\[
e_r\ge\tfrac12\{\tau(|T|^2R)-\tau(|T|^2S_1)\}
\ge\tfrac g2\tau(R(I-S_1)).
\]
Here the last inequality follows by splitting the trace over $R$ and $R^\perp$ and using $\tau(R(I-S_1))=\tau(R^\perp S_1)$. Hence $\|R-S_1\|_{2,\tau}^2\le4e_r/g$. The range projection formula gives
\[
S_2=TS_1(S_1T^*TS_1)^{-1}S_1T^*,
\]
with the inverse taken in the corner. Since $(I-L)TR=0$, we have
\[
\tau((I-L)S_2)\le\|T^{-1}\|^2\|T\|^2\tau(R^\perp S_1).
\]
Using $\|T^{-1}\|\le(pq)^{-1}$ and the equal traces of $L,S_2$, we get
\begin{equation}\label{eq:outputgap}
\|L-S_2\|_{2,\tau}^2\le\frac{4e_r}{p^2q^2g}.
\end{equation}
The projection $P$ is an upper spectral projection for $D$. Since $D^2=I+2U+U^2-2V+V^2$ and $0\le V\le I$, \eqref{eq:volumeloss} gives
\[
\tau(U(I-S_2))+\tau(VS_2)\le2M^2e_r.
\]
For example, the trace difference in \eqref{eq:volumeloss} is the sum of $\tau((2U+U^2)(I-S_2))$ and $\tau((2V-V^2)S_2)$; the identity terms cancel because the projections have equal trace. All these traces are finite.
By the triangle inequality in $L^2(\M,\tau)$,
\[
\tau(U(I-L))\le2\tau(U(I-S_2))+2M\|L-S_2\|_{2,\tau}^2.
\]
The same estimate holds for $\tau(VL)$ with $\tau(VS_2)$ in place of $\tau(U(I-S_2))$. Together with \eqref{eq:outputgap} this bounds the sum of the first two terms of \eqref{eq:leak} by $(4M^2+16M/(p^2q^2g))e_r$. Applying the argument to the adjoint product $ET^*D$ interchanges $D,E$ and $L,R$ and leaves the excess unchanged. The same choice of tolerance $r$ is used for the adjoint product. Adding the two bounds and letting $r\downarrow0$ proves the lemma.
\end{proof}

\begin{theorem}\label{thm:minimizers}
Under \eqref{eq:gap}, a pair in $\mathcal C_b$ attains the minimum in Theorem~\ref{thm:flatten} if and only if
\begin{equation}\label{eq:minimizers}
U=J_1UJ_1,\quad V=J_2VJ_2,\quad W=J_3WJ_3,\quad Z=J_4ZJ_4,
\end{equation}
where
\[
J_1=P\wedge L,\quad J_2=P^\perp\wedge L^\perp,\quad
J_3=Q\wedge R,\quad J_4=Q^\perp\wedge R^\perp.
\]
\end{theorem}
\begin{proof}
If $\varepsilon=0$, Lemma~\ref{lem:leak} and faithfulness of the trace imply 
\[(I-L)U^{1/2}=LV^{1/2}=(I-R)W^{1/2}=RZ^{1/2}=0. \]
The support conditions of $U,V,W,Z$ now give \eqref{eq:minimizers}.
Conversely, these conditions imply that $L$ reduces $D$, with $DL\ge L$ and $DL^\perp\le L^\perp$. Similarly $R$ reduces $E$, with $ER\ge R$ and $ER^\perp\le R^\perp$. Thus $L,R$ are upper spectral projections for $D,E$, respectively. Applying \eqref{eq:volumeadd} to $R$ yields
$v_{DTE}(R)=\ell_t(E)+\ell_t(T)+\ell_t(D)$.
The volume bound and \eqref{eq:submult} imply equality in \eqref{eq:excess}.
\end{proof}

Let $\mathcal E_b$ denote the set of minimizers in $\mathcal C_b$. For the quantitative estimate we assume that, for some $\gamma>0$,
\begin{equation}\label{eq:geometry}
\begin{aligned}
P(I-L)P&\ge\gamma(P-J_1),& P^\perp LP^\perp&\ge\gamma(P^\perp-J_2),\\
Q(I-R)Q&\ge\gamma(Q-J_3),& Q^\perp RQ^\perp&\ge\gamma(Q^\perp-J_4).
\end{aligned}
\end{equation}
In finite dimension such a constant always exists. In general it is a separate assumption.

\begin{lemma}\label{lem:compress}
Let $H,F$ be projections, $J=H\wedge F$, and $H(I-F)H\ge\gamma(H-J)$. If $X=HXH\ge0$ and $\|X\|\le M$, then
\[
\|X-JXJ\|_{2,\tau}^2\le\frac{2M}{\gamma}\tau(X(I-F)).
\]
\end{lemma}
\begin{proof}
Put $K=H-J$. The three blocks $JXK,KXJ,KXK$ are orthogonal for the squared $L^2$ norm. Consequently,
\[
\|X-JXJ\|_{2,\tau}^2\le2\tau(KX^2K)\le2M\tau(KXK).
\]
The assumed operator inequality gives $\gamma\tau(KXK)\le\tau(X^{1/2}(I-F)X^{1/2})$. The argument holds also for extended nonnegative traces, and proves the assertion.
\end{proof}

\begin{theorem}\label{thm:stability}
Assume \eqref{eq:gap} and \eqref{eq:geometry}. For $(A,B)\in\mathcal C_b$ with $\|D\|,\|E\|\le M$, one has
\begin{equation}\label{eq:stability}
\varepsilon\ge\frac{\gamma}{16M^3+64M^2/(p^2q^2g)}\operatorname{dist}_{2,\tau}((A,B),\mathcal E_b)^2,
\end{equation}
where the distance is the infimum of $(\|A-A'\|_{2,\tau}^2+\|B-B'\|_{2,\tau}^2)^{1/2}$ over $(A',B')\in\mathcal E_b$.
\end{theorem}
\begin{proof}
Set
\[
D_*=I+J_1UJ_1-J_2VJ_2,\qquad E_*=I+J_3WJ_3-J_4ZJ_4.
\]
Since $\|V\|,\|Z\|<1$, these operators are positive and invertible. Put $A_*=D_*A_0$ and $B_*=B_0E_*$. The pair $(A_*,B_*)$ belongs to $\mathcal C_b$. By Theorem~\ref{thm:minimizers}, it belongs to $\mathcal E_b$. Apply Lemma~\ref{lem:compress} to the four deviations, using \eqref{eq:geometry}. The deviations on complementary corners are orthogonal in $L^2$. Hence Lemma~\ref{lem:leak} gives
\[
\|D-D_*\|_{2,\tau}^2+\|E-E_*\|_{2,\tau}^2
\le\frac{2M}{\gamma}\left(8M^2+\frac{32M}{p^2q^2g}\right)\varepsilon.
\]
In particular, the differences belong to $L^2$, even if the operators themselves do not. Since $A_0,B_0$ are contractions, the same bound holds for the squared distance to $(A_*,B_*)$. This proves \eqref{eq:stability}.
\end{proof}

\begin{remark}
The constant in \eqref{eq:stability} is independent of $t$. The geometric assumption cannot be inferred from the spectral gap alone: on generic two-dimensional fibers with $p\ne q$, the upper left and right singular projections converge to $P$ and $Q$ as the principal angle tends to zero. If nonzero angles accumulate at zero, the positive spectra of the compressions in \eqref{eq:geometry} may accumulate at zero. We do not claim that the constant in \eqref{eq:stability} is optimal.
\end{remark}

\begin{corollary}\label{cor:sharp}
The exponent $2$ in \eqref{eq:stability} cannot in general be replaced by any smaller positive exponent, even locally near a minimizer.
\end{corollary}
\begin{proof}
In $M_3(\mathbb C)$ with the usual trace, take $t=2$, $p=q=1/2$, $P=\operatorname{diag}(1,1,0)$ and
\[
Q=\begin{pmatrix}1&0&0\\0&1/2&1/2\\0&1/2&1/2\end{pmatrix}.
\]
Then $T=1\oplus\begin{pmatrix}3/4&1/4\\1/8&3/8\end{pmatrix}$. The squared singular values of the second block are $(25\pm3\sqrt{41})/64$, so \eqref{eq:gap} holds. The upper left and right singular vectors of the second block are neither parallel nor orthogonal to the corresponding ranges of $P,Q$ in that block (as follows directly from the two Gram matrices). Thus $J_1=J_3=e_1e_1^*$ and $J_2=J_4=0$. By Theorem~\ref{thm:minimizers},
\[
\mathcal E_b=\{(A_0+u e_1e_1^*,B_0+v e_1e_1^*):u,v\ge0\}.
\]
For real $h$, put $z_h=e_1+he_2$, $D_h=I+z_hz_h^*$, $A_h=D_hA_0$, and $B_h=B_0$. These pairs are admissible and $\|D_h\|\le3$ for $|h|\le1$. Direct minimization over $u,v$ gives
\begin{equation}\label{eq:sharpdist}
\operatorname{dist}_{2,\tau}((A_h,B_h),\mathcal E_b)^2=2h^2+h^4.
\end{equation}
Let $\alpha=\langle Le_2,e_2\rangle=(1+5/\sqrt{41})/2$. Since $D_h^2=I+(3+h^2)z_hz_h^*$ and $Le_1=e_1$, the determinant of its compression to $L$ is $1+(3+h^2)(1+\alpha h^2)$. Also $\ell_2(D_h)=\log(2+h^2)$. The volume bound on $D_hT$ evaluated at $R$, together with \eqref{eq:volumeadd}, yields
\[
0\le\varepsilon_h\le\ell_2(D_h)-v_{D_h}(L)
=-\tfrac12\log\left(1-\frac{(3+h^2)(1-\alpha)h^2}{(2+h^2)^2}\right).
\]
The upper bound is $O(h^2)$ as $h\to0$. By \eqref{eq:stability} and \eqref{eq:sharpdist}, $\varepsilon_h$ is also bounded below by a positive constant times $h^2$. For every $0<\beta<2$, the ratio of $\varepsilon_h$ to the $\beta$-th power of the distance in \eqref{eq:sharpdist} tends to zero. This proves the assertion.
\end{proof}

\section{Applications to Gibbs states}
Let $H,K$ be self-adjoint operators affiliated with $\M$ and bounded from below. Suppose that $e^{-\beta H},e^{-\gamma K}\in L^1(\M,\tau)$ for some $\beta,\gamma>0$. Define the Gibbs states
\[
\rho=\frac{e^{-\beta H}}{\tau(e^{-\beta H})},\qquad
\sigma=\frac{e^{-\gamma K}}{\tau(e^{-\gamma K})}.
\]
Let $P,Q$ be reducing spectral projections of $H,K$, respectively, with $0<\tau(P)=\tau(Q)=t<\infty$. Assume that
\[
\begin{aligned}
HP&\le h_-P,& HP^\perp&\ge h_+P^\perp,\\
KQ&\le k_-Q,& KQ^\perp&\ge k_+Q^\perp,
\end{aligned}
\]
where $h_-<h_+$ and $k_-<k_+$. These inequalities are understood on the corresponding reducing subspaces in the quadratic-form sense.

\begin{corollary}\label{cor:gibbs}
For $u,v>0$, put $p=e^{-u\beta(h_+-h_-)}$ and $q=e^{-v\gamma(k_+-k_-)}$. Then
\begin{equation}\label{eq:gibbs}
\delta_t(\rho^u,\sigma^v)\ge -\frac{t}{2}\int_{[0,1]}\log\eta_{p,q}(x)\,d\nu_{P,Q}(x)
\ge \frac{C(p,q)}4\|P-Q\|_{2,\tau}^2.
\end{equation}
In particular, the defect is zero if and only if $P=Q$.
\end{corollary}
\begin{proof}
Write $Z_H=\tau(e^{-\beta H})$ and $Z_K=\tau(e^{-\gamma K})$. Set
\[
A=Z_H^u e^{u\beta h_-}\rho^u,\qquad B=Z_K^v e^{v\gamma k_-}\sigma^v.
\]
By spectral calculus, $AP\ge P$ and $AP^\perp\le pP^\perp$. Similarly, $BQ\ge Q$ and $BQ^\perp\le qQ^\perp$. Both operators are bounded and belong to $L_{\log_+}(\M,\tau)$. Thus $(A,B)\in\mathcal C$. Scalar invariance of the defect and Theorems~\ref{thm:flatten} and \ref{thm:exact} give \eqref{eq:gibbs}.

The assertion about zero defect follows from Theorem~\ref{thm:flatten}.
\end{proof}

\begin{example}
Let $\M=\mathcal B(\ell^2(\mathbb N_0))$
with its usual trace, and let $(e_n)_{n\ge0}$
be the standard orthonormal basis.
For $\Delta>0$, define the positive self-adjoint
operator $H$ by $He_n=n\Delta e_n$, with domain
\[
D(H)=\left\{\xi=(\xi_n)\in\ell^2(\mathbb N_0):
\sum_{n=0}^{\infty}n^2|\xi_n|^2<\infty\right\}.
\]
Let $U_\theta$ be the rotation through
$\theta\in[0,\pi/2]$ on
$\operatorname{span}\{e_0,e_1\}$ and the identity
on its orthogonal complement.
Set $K=U_\theta HU_\theta^*$.
Both Hamiltonians are unbounded.
For $\beta,\gamma>0$, their Gibbs operators
are trace class since
\[
\Tr(e^{-\beta H})=(1-e^{-\beta\Delta})^{-1},
\qquad
\Tr(e^{-\gamma K})=(1-e^{-\gamma\Delta})^{-1}.
\]
Let $P$ be the projection onto $\mathbb C e_0$
and $Q=U_\theta PU_\theta^*$.
Then $t=1$,
$PQP=\cos^2\theta P$, and
$\|P-Q\|_{2,\Tr}^2=2\sin^2\theta$.

For $u,v>0$, set
$p=e^{-u\beta\Delta}$,
$q=e^{-v\gamma\Delta}$,
$A=e^{-u\beta H}$ and $B=e^{-v\gamma K}$.
By scalar invariance,
\[
\delta_1(\rho^u,\sigma^v)=\delta_1(A,B).
\]
On $\operatorname{span}\{e_0,e_1\}$,
the operators $A$ and $B$ coincide with the
corresponding two-dimensional flattened operators.
Hence the norm of their product on this subspace
is $\sqrt{\eta_{p,q}(\cos^2\theta)}$.
On the orthogonal complement,
$ABe_n=(pq)^n e_n$ for $n\ge2$.
Since
\[
(pq)^2<\max\{p,q\}
\le\sqrt{\eta_{p,q}(\cos^2\theta)},
\]
we obtain
\[
\|AB\|=\sqrt{\eta_{p,q}(\cos^2\theta)}.
\]
As $\|A\|=\|B\|=1$, it follows that
\[
\delta_1(\rho^u,\sigma^v)
=-\tfrac12\log\eta_{p,q}(\cos^2\theta)
\ge\tfrac12 C(p,q)\sin^2\theta.
\]
Since $\eta_{p,q}(1)=1$ and
$\eta_{p,q}'(1)=C(p,q)$, we have
$-\log\eta_{p,q}(\cos^2\theta)
=C(p,q)\sin^2\theta+O(\sin^4\theta)$.
Consequently,
\[
\lim_{\theta\to0^+}
\frac{\delta_1(\rho^u,\sigma^v)}
{\tfrac12 C(p,q)\sin^2\theta}=1.
\]
Thus the lower bound is asymptotically sharp
as $\theta\to0$.

This example shows that the defect controls
the rotation of the ground-state projection
even for unbounded Hamiltonians.
Here $K-H$ extends to a bounded finite-rank
operator, since the rotation acts only on
the first two energy levels.
\end{example}


\begin{thebibliography}{9}
\bibitem{DoddsEtAl} P.~G. Dodds, T.~K. Dodds, F.~A. Sukochev, and D. Zanin, Logarithmic submajorization, uniform majorization and H\"older type inequalities for $\tau$-measurable operators, Indag. Math. (N.S.) \textbf{31} (2020), 809--830.
\bibitem{FackKosaki} T. Fack and H. Kosaki, Generalized $s$-numbers of $\tau$-measurable operators, Pacific J. Math. \textbf{123} (1986), 269--300.
\bibitem{FugledeKadison} B. Fuglede and R. V. Kadison, Determinant theory in finite factors, Ann. of Math. (2) \textbf{55} (1952), 520--530.
\bibitem{HalmosTwoSubspaces} P.~R. Halmos, Two subspaces, Trans. Amer. Math. Soc. \textbf{144} (1969), 381--389.
\end{thebibliography}
\end{document}